\documentclass[]{article}
\usepackage{fullpage}
\usepackage{amsmath}
\usepackage{amsthm}
\usepackage{graphicx}
\usepackage[round]{natbib} 
\usepackage[colorlinks=true,citecolor=blue,linkcolor=black,urlcolor=blue]{hyperref}
\usepackage{multirow}
\usepackage{float,epsf,caption,subcaption}
\usepackage{longtable}
\usepackage{tabularx}
\usepackage{authblk}
\usepackage{amsfonts}

\newtheorem{theorem}{Theorem}[section]

\newtheorem{definition}[theorem]{Definition}

\title{Trolley Optimisation for Loading Printed Circuit Board Components}

\author[1,2]{Vinod Kumar Chauhan\thanks{vinod.kumar@eng.ox.ac.uk (corresponding author). Accepted to \textbf{Operations Research Forum} (Springer Nature).}}
\author[3]{Mark Bass\thanks{Mark.Bass@Rolls-Royce.com}}
\author[1]{Ajith Kumar Parlikad\thanks{aknp2@cam.ac.uk}}
\author[1]{Alexandra~Brintrup\thanks{ab702@cam.ac.uk}}
\affil[1]{Department of Engineering, University of Cambridge UK}
\affil[2]{Department of Engineering Science, University of Oxford UK}
\affil[3]{Rolls-Royce}

\begin{document}

	\maketitle
	
\begin{abstract}
    A trolley is a container for loading printed circuit board (PCB) components and a trolley optimisation problem (TOP) is an assignment of PCB components to trolleys for use in the production of a set of PCBs in an assembly line.
    In this paper, we introduce the TOP, a novel operation research application. To formulate the TOP, we derive a novel extension of the bin packing problem. We exploit the problem structure to decompose the TOP into two smaller, identical and independent problems. Further, we develop a mixed integer linear programming model to solve the TOP and prove that the TOP is an NP-complete problem.
    A case study of an aerospace manufacturing company is used to illustrate the TOP which successfully automated the manual process in the company and resulted in significant cost reductions and flexibility in the building process. 


\end{abstract}
	
\section{Introduction}
\label{sec_intro}

Printed circuit boards (PCBs) are very important electronic components due to their usage in manufacturing electrical and electronic devices ranging from remote controls, mobile phones, laptops and televisions to train and aircraft engines etc. The problems related to PCBs have been studied extensively because of their importance, e.g., \cite{drezner1984optimizing}. Due to the growing need for PCBs, PCB manufacturers use automated and complex manufacturing systems, like collect-and-place (CAP) machines, (\cite{lin2015demand}). Typically, manufacturers have assembly shops with multiple assembly lines, each may have multiple CAP machines. So, automated setups are needed to efficiently utilise the machines which are otherwise very expensive (\cite{wu2009scheduling}).

PCB assembly planning is a multi-level optimisation problem which consists of several interdependent problems (refer to Figs.~2~and~3 in \cite{mumtaz2019multi} and \cite{ji2001planning}). Each of the problems in the PCB assembly planning is an NP-hard problem. The complexity of these problems is exacerbated by their large scale, involving a large number of components and a diverse range of PCBs, making it challenging to comprehensively address all sub-problems and achieve an integrated solution for optimal outcomes (\cite{gao2021hierarchical}). In fact, despite being studied for several decades (\cite{drezner1984optimizing,ahmadi1988component}), even individual machine-level problems are solved approximately using heuristic methods (\cite{li2018clustering}).

A trolley and a stacker are containers used to hold components, such as registers, capacitors and diodes, required to build PCBs. Instead of plugging individual components directly into CAP machines, i.e., assigning individual components to machine slots (\cite{gao2021hierarchical}), components are loaded onto trolleys and stackers which then are plugged into the CAP machines. This helps to manage components, which are large in number and also helps to change components while building a variety of PCBs. So, the TOP is a problem of finding a minimum number of trolleys and stackers with common capacities required to load a given set of PCB components of different sizes to build a given set of PCBs in an assembly line with multiple machines. {\color{blue}Key terms used in defining the TOP are defined in Table~\ref{tab_terminology}.}

\begin{table}[htb!]
\caption{Application specific terminology}
\label{tab_terminology}
\centering
    \begin{tabular}{m{3.5cm}m{11cm}}
    \hline
    \textbf{Term}          & \textbf{Definition}   \\ \hline
    PCB                    & A Printed Circuit Board (PCB) is a crucial electronic component used in making various electrical devices, such as remote controls, mobile phones, laptops, and televisions, as well as in larger systems like train and aircraft engines. \\ \hline 
    Component              & These are the basic units required to manufacture PCBs, e.g., registers, capacitors and diodes.                                                                                                                                            \\ \hline 
    Trolley                & A trolley is a container for loading printed circuit board (PCB) components, having 33 slots of equal size to plug components.                                                                                                         \\ \hline 
    Stacker                & A stacker (also called auto-sequencer) is a container for loading large PCB components, having 30 slots of equal size to plug components. \\ \hline  
    TOP   & TOP is a problem of finding a minimum number of trolleys and stackers with common capacities required to load a given set of PCB components of different sizes to build a given set of PCBs in an assembly line with multiple machines. \\ \hline 
    CAP Machine            & A collect-and-place (CAP) machine is a type of automated equipment used in manufacturing and assembly processes. It is designed to pick up components from a trolley/stacker and place them onto a PCB. \\ \hline 
    Assembly line          & An assembly line is a series of CAP machines where PCBs are sequentially assembled as they progress through each machine in the line. \\ \hline 
    Assembly shop          & An assembly shop is a facility to manufacture PCBs. It often includes assembly lines, testing stations, and quality control areas.  \\ \hline 
    Assembly line capacity & It is the maximum number of containers (trolleys or stackers) that can be used in an assembly line. It is referred to as the maximum trolley limit in the model development. \\ \hline                            
    \end{tabular}
\end{table}

To the best of our knowledge, this is the first study of the TOP as the existing literature assumes that components are directly plugged into slots on CAP machines (\cite{castellani2019printed,gao2021hierarchical}). In practice, PCB manufacturers use trolleys to load components which otherwise could be difficult to manage and switch while building a variety of PCBs on an assembly line. If trolleys are prepared from direct assignment of components to slots of CAP machines, in that case, the trolley loading won't be efficient, as that may need to switch multiple trolleys between different PCBs. The TOP is an important problem, especially for low-volume and high-mix problems, due to the frequent need to switch PCBs.

To formulate the TOP, we extend the bin packing problem (BPP) which finds a minimum number of bins of common capacity to pack a given set of items of different weights (\cite{wascher2007improved}). The TOP shares constraints similar to the BPP, with additional constraints (for details refer to Subsection~\ref{subsec_problem_formulation}) to ensure that the number of trolleys needed to build each PCB does not exceed the capacity of the assembly line otherwise the problem will be either infeasible or will need to change trolleys while building a PCB. We also exploit the problem structure to decompose the TOP into two smaller, identical and independent problems, i.e., assignment of components to trolleys and assignment of components to stackers, by pre-computing the dependency between both problems. Further, we develop a single and a smaller mixed integer linear programming (MILP) model to solve both problems. The exact optimisation-based methods are used to successfully solve the resulting MILP model.

This paper also presents a case study of an aerospace manufacturing company. The company has an assembly shop with three assembly lines with one, two and three collect-and-place (CAP) machines, respectively, to meet the PCB demands of all its products. Due to the complexity associated with the PCB planning (\cite{koskinen2020rolling}) and due to the customised needs of the company, the company can't use the CAP machine software to perform the setups. Currently, experienced managers handle all setups, including the TOP. The TOP is the most time-consuming setup task and takes eight weeks to perform the setups. The longer setup time results in a loss of flexibility in the build process to introduce new PCBs and to deal with unprecedented situations, like COVID-19, which might need rescheduling production due to changed industrial requirements.

The contributions of the study are summarised as given below.
\begin{itemize}
	
    \item A novel trolley optimisation problem (TOP) is introduced which finds a minimum number of trolleys and stackers of common capacities to load a given set of PCB components of different sizes, required to manufacture a set of PCBs in an assembly line.
    
    \item A novel extension of the BPP is derived to formulate the TOP by introducing additional constraints to ensure that the number of trolleys required to build each PCB is less than or equal to the capacity of the assembly line used for building the PCB. An MILP model is developed to solve the TOP which is solved using exact optimisation methods.
    
    \item The problem structure is exploited to decompose the TOP into two smaller, identical and independent problems, i.e., assignment of trolleys and assignment of stackers, by pre-computing the dependency between them. So, a single and smaller MILP model is sufficient to solve both the problems and, hence, to solve the TOP (for details refer to Subsection~\ref{subsec_problem} and \ref{subsec_solution_approach}).
    
    \item A theoretical proof is provided to prove that the TOP is an NP-complete.
    
    \item An industrial case study of a company is presented that currently uses a manual setup for the TOP and suffers from long setup times and lack of flexibility to introduce new products and to respond to unprecedented situations, like COVID-19. The proposed model helps the company to automate the TOP and get rid of issues caused by the manual process. Moreover, the model is deployed in the company and resulted in significant cost reductions through automation of the TOP, in addition to introducing flexibility in the building process.
	
\end{itemize}

The rest of this paper is structured as follows. In Section~\ref{sec_literature}, we present related work on PCB optimisation and the BPP to situate the context of our contribution. In Section~\ref{sec_case_study}, we characterise the TOP, formulate it using a case study of a company, present comparative results against the existing benchmark, and finally derive a proof in Section~\ref{sec_np_complete} to show that the TOP is NP-complete. We then conclude and discuss implications and managerial insights in Section~\ref{sec_conclusion}.

\section{Related work}
\label{sec_literature}
In this section, we present a brief literature on the BPP and the problems related to PCBs.

\subsection{The bin packing problems}
\label{subsec_bin_packing}

The bin packing problems (BPPs) are one of the most popular NP-hard combinatorial problems, which have been studied since the thirties (\cite{kantorovich1960mathematical}). A large number of variants of the BPP have been derived with a variety of applications across different areas, like the transportation industry, for example, delivery of parcels, pallet loading, container loading, cutting stock problem and large scale integration placement problem (\cite{delorme2018bpplib}).

According to the dimensions, the problem has one-, two- and multi-dimensional variants. For example, one, two and three-dimensional cutting stock problems, which minimise the wastage of material by cutting the stock material, like paper/fabric roll and metal sheet, into pieces of specified sizes (\cite{martello2000three}).

In terms of objectives, the BPPs can be single- and multi-objective problems. For example, \cite{aydin2020multi} applied the multi-objective BPP to cloud computing to tackle the virtual machine placement problem. They considered two objectives minimising the number of servers and the number of fire-ups of idle servers. Similarly, depending on the nature of the problem, the BPP can be categorised as linear, non-linear, constraint programming and integer programming problems (\cite{leao2020irregular}).

For solving the BPPs, a variety of algorithms have been studied, like exact methods and heuristic methods (\cite{delorme2016bin}). The heuristics can be further categorised as approximation algorithms, meta-heuristics and hyper-heuristics. The approximation algorithms, like first fit, next fit and worst fit, do not provide optimal solutions but provide guaranteed bounds on the solution. The meta-heuristics, like ant colony optimisation, particle swarm optimisation and genetic algorithms etc., also do not provide any guarantees on optimality but can provide quicker solutions to complex problems. The hyper-heuristics combine different heuristic techniques to solve complex problems (\cite{burke2013hyper}).

Please refer to \cite{wascher2007improved} for the BPP typology, \cite{delorme2018bpplib} for a library of different BPP and \cite{delorme2016bin} for a detailed review of mathematical models and exact algorithms to solve the BPP.

\subsection{PCB related problems}
\label{subsec_pcb}

PCB production/assembly planning is related with the configuration of different assembly lines and CAP machines in an assembly shop for scheduling a given demand of PCBs. This is, mainly, a three-level optimisation problem, involving machine-level, assembly line-level and shop-level tasks (refer to Figs.~2~and~3 in \cite{mumtaz2019multi} and \cite{ji2001planning}). Shop-level involves problems related to the entire assembly shop, like PCB assignments to different assembly lines. Assembly line-level considers all machines in an assembly line and involves problems, like job sequencing, where the sequence of PCB manufacturing is decided, and line balancing, where components are allocated to different machines of an assembly line. Machine-level involves problems related to a single machine like placement sequence of components, assignment and sequence of nozzles, nozzle change schedule, and assignment of components to feeder slots. Machine-level problems are dependent on the type of machine used, like single-head or multi-head gantry machines, and need their own optimisation solutions thereby increasing the complexity. Please refer to \cite{ji2001planning} for a review of different PCB-related optimisation problems.

All the sub-problems, relating to the three levels, are interdependent NP-hard large-scale problems (because of a large number of parts to be placed on the boards and a large variety of PCBs) which makes it difficult to find the optimal solution for the integrated problem. In fact, despite being studied for several decades (\cite{drezner1984optimizing,ahmadi1988component}), even individual machine-level problems are solved approximately using heuristic methods. For example, \cite{li2018clustering} solved machine-level tasks using clustering based heuristic method, assuming line and shop-level problems are already solved. Due to the complexity, these sub-problems are either solved separately or some of them are solved together, and use a variety of approaches to solving them, e.g., simulation (\cite{stayer2011simulation}), heuristic (\cite{koskinen2020rolling}), meta-heuristic (\cite{qin2019two} and approximation methods (\cite{ellis2002optimization}).

Existing literature, mainly, considers component loading problems as feeder slot assignment problems at the machine level, which assigns components to slots on the machine. The problem assumes that the number of feeder slots is sufficient to load different components which might not be true while working with a variety of PCBs. Here, we discuss some of the recent work to deal with the component loading problem. \cite{castellani2019printed} studied a moving-board-with-time-delay type machine for minimisation of assembly time of PCBs, considering machine-level sub-problems of feeder slot assignments and placement sequences. They applied the Bees Algorithm to solve the problem. \cite{mumtaz2019multi} studied sub-problems of assembly line assignment to PCBs, component allocation to machines and component placement sequencing and solved the model using a hybrid spider monkey optimisation algorithm. \cite{koskinen2020rolling} considered sub-problems of PCB allocation to assembly lines, load balancing of machines of a line and job scheduling and developed a two-phase heuristic to solve the problem. \cite{gao2021hierarchical} considered a beam-head type machine to optimise machine-level sub-problems of nozzle assignment, feeder slot assignments and placement sequence optimisation. They developed a hierarchical heuristic to solve the sub-problems in descending order of their importance.

In practice, PCB manufacturers use trolleys to load components which otherwise could be difficult to manage and switch while building a variety of PCBs in an assembly line. If trolleys are prepared from direct assignment of components to slots of CAP machines, in that case, the trolley loading won't be efficient, as that may need to switch multiple trolleys between different PCBs. Thus, the TOP is a very important problem, especially for low-volume and high-mix problems, due to the frequent need to switch PCBs.

From this brief literature review, it is clear that the TOP -- a novel extension of the BPP -- to the best of our knowledge is not studied as such in the PCB literature but taken care of indirectly through slot assignment problem which does not provide flexibility of managing a large number of components, especially in problems with high-mix, i.e., a large variety of PCBs.

\section{Trolley optimisation: a case study}
\label{sec_case_study}
In this section, we characterise the problem, develop a mathematical formulation and present comparative results against the existing benchmark.

\subsection{Problem definition}
\label{subsec_problem}
An original equipment manufacturer in aerospace technology has a PCB assembly shop to meet the demands of all its products. The assembly shop has three assembly lines with one, two and three collect-and-place (CAP) machines, respectively. This requires the assembly shop to operate in a low-volume and high-mix setting and in a unique way. So, first, due to the complexity of the problem and second, due to the unique operation of the assembly shop, the company can't use the CAP software for PCB planning and optimisation. So, except for the machine-level tasks, like placement sequences and nozzle changes etc., all other PCB planning tasks, like the TOP are performed by an experienced manager which takes more than eight weeks each time an assembly line needs to be set up. 

Due to a long manual setup time, the company lacks the flexibility to introduce new PCBs in the build process as well as in responding to unprecedented situations, like COVID-19, for which the company may have to reschedule the build process according to the changed demands. Moreover, the manual setup also leads to a suboptimal solution, requiring more trolleys to load PCB components which are difficult to manage. To address these issues, we automate the TOP in this case study, as discussed below.

The company builds $P$ different types of PCBs using $C$ different types of components. The components can be categorised into two categories, (i) the components which need one to five slots on a container for loading different PCB components onto the CAP machine, and are put on the trolleys, and (ii) the components needing more than five slots on the container and are loaded onto the stackers. The stackers are used for larger components due to placement efficiency. So, the components are loaded onto trolleys and stackers which are then loaded onto CAP machines. Thus, instead of loading individual components directly onto the CAP machine, which could be a very complex process, trolleys and stackers are loaded onto the CAP machine. This makes the task of loading different components for different PCBs easier and more manageable during the operation of the assembly shop (see Fig.~\ref{fig_aerospace} for an example of trolleys). Moreover, each component in an assembly line is loaded onto only one trolley/stacker for easy traceability. Each assembly line has a predefined capacity for the number of trolleys and stackers that it can accommodate. Additionally, it's important to note that both trolleys and stackers occupy an equal amount of space within the assembly line configuration.

\begin{figure*}[htb!]
	\centering
	\begin{subfigure}{0.5\textwidth}
		\centering
		\includegraphics[width=\textwidth]{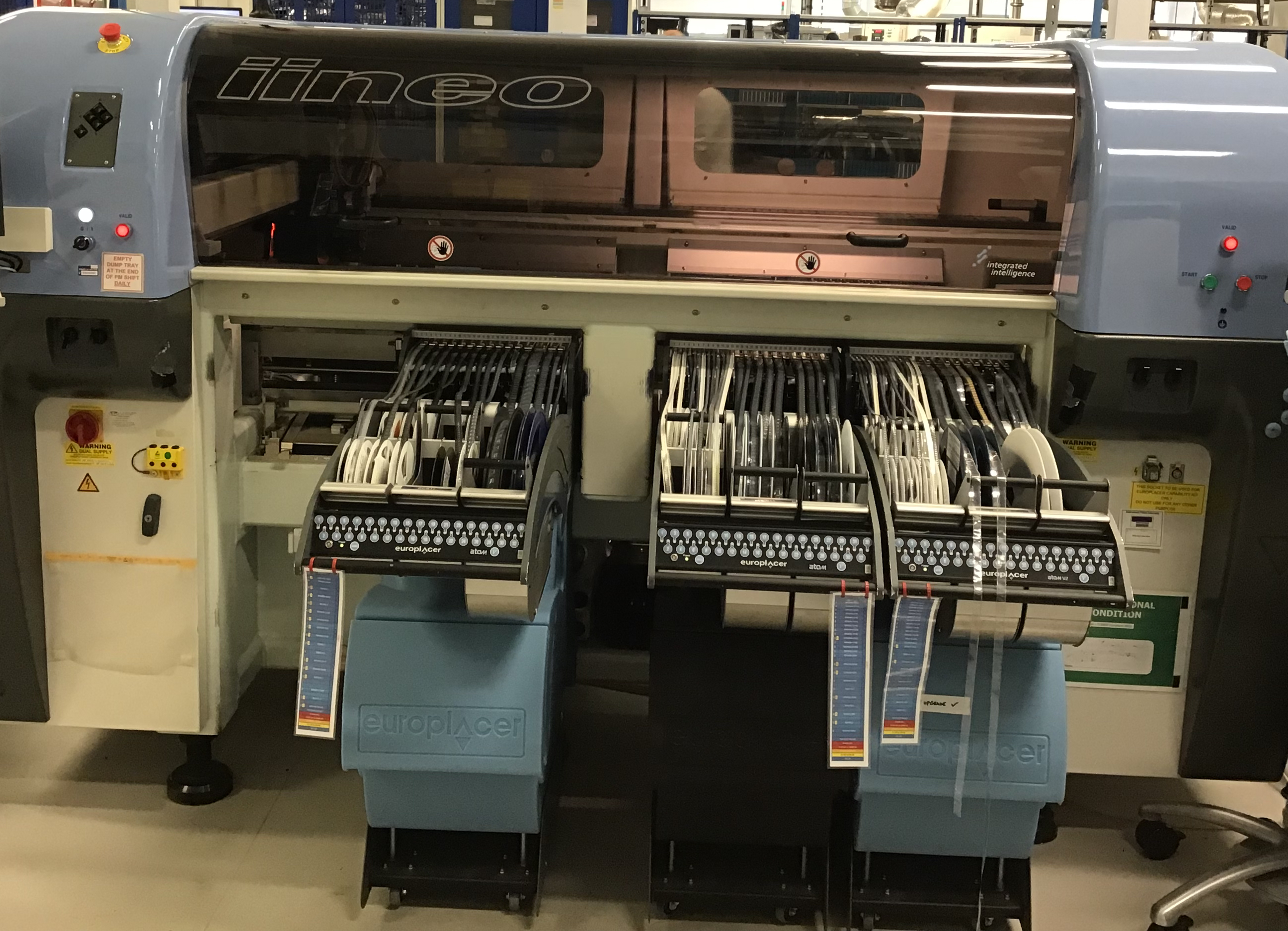}
		\caption{An example of a CAP machine where one side is shown with three of the four places occupied by trolleys.}
		\label{subfig_CAP}
	\end{subfigure}%
	~~
	\begin{subfigure}{0.5\textwidth}
		\centering
		\includegraphics[width=\textwidth]{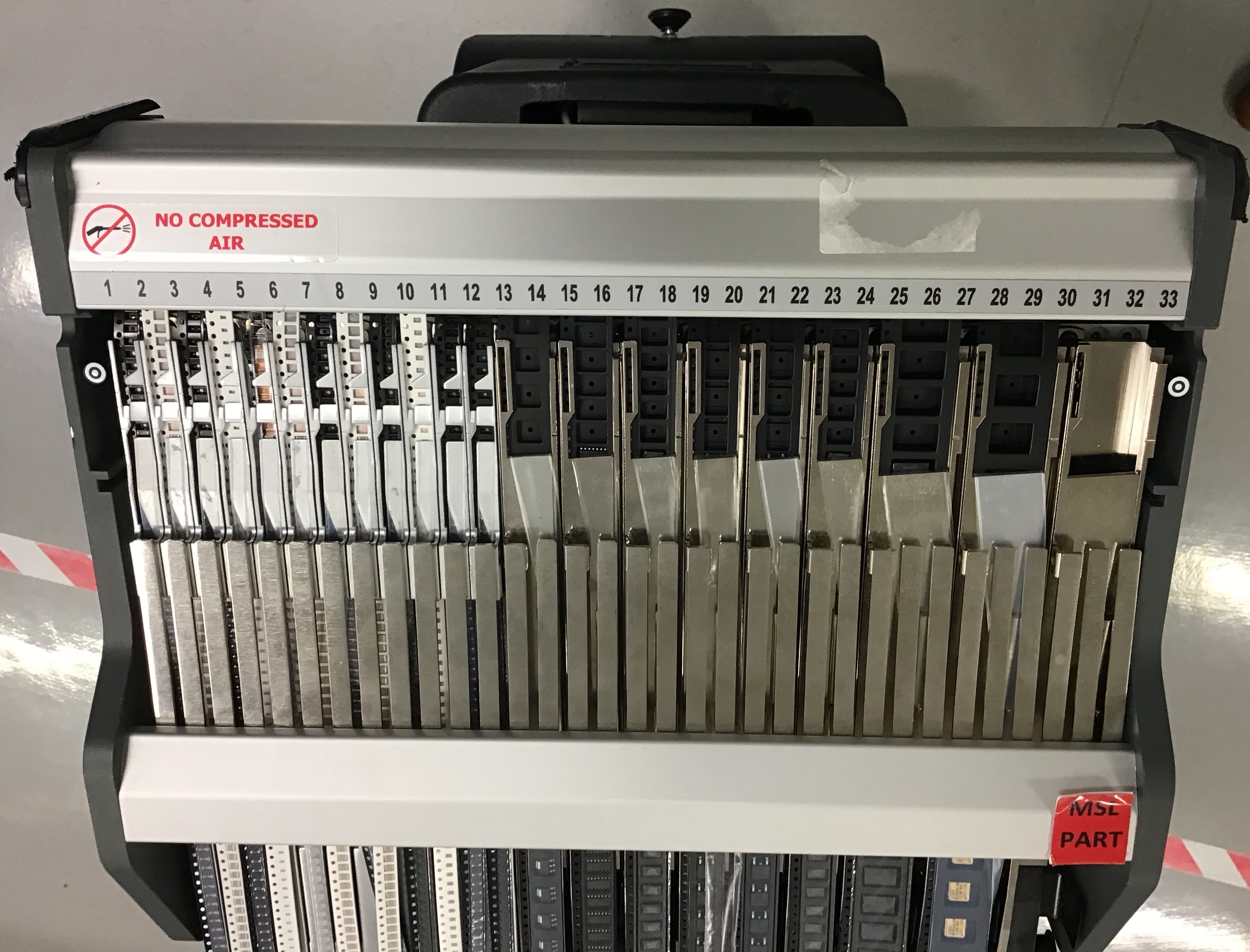}
		\caption{Slots of a trolley where components are plugged and each component can take one to five slots depending on the size.}
		\label{subfig_trolley}
	\end{subfigure}	
	\caption{An example of trolleys, slots on the trolley, and a CAP machine (pictures are shared by our industrial partner).}
	\label{fig_aerospace}
\end{figure*}

Thus, the TOP aims to find a minimum number of trolleys and stackers of common capacities to load a given set of PCB components of varying sizes (measured by the number of slots required) to build a set of PCBs in an assembly line, subject to the constraint that the total number of trolleys and stackers used for each PCB must not exceed the assembly line capacity (referred as maximum trolley limit). The assignment of components to trolleys and to stackers are exactly similar and are solely dependent on the assembly line capacity. For instance, suppose a PCB is built on an assembly line with a capacity of 16 trolleys. If a PCB requires two stackers, the remaining components must be loaded onto 14 trolleys, as two spaces are occupied by stackers (a stacker and a trolley occupy equal space on the assembly line). This information can be determined in advance based on the PCBs to be manufactured as we can precompute if we need stackers. Leveraging this structural characteristic, the problem is decomposed into two smaller and independent sub-problems: trolley assignment and stacker assignment. Consequently, a single, smaller MILP model is developed to address both assignments. For simplicity, we primarily focus on the trolley assignment (refer to Subsection \ref{subsec_problem_formulation}), as the stacker assignment is trivial due to the manageable scale of the problem.

In the development of a model for the problem, we have made the following assumptions.
\begin{itemize}
    \item Capacity of each trolley is the same, i.e., 33 slots and that of each stacker is 30 slots.
    \item Each component is loaded onto only one trolley/stacker, for easy traceability, in an assembly line.
    \item Different trolley loading setups do not make a significant difference in the operation of the assembly line and are not relevant to this problem.
\end{itemize}

\subsection{Problem formulation}
\label{subsec_problem_formulation}
The objective of the TOP is to minimise the number of trolleys needed to load all the PCB components. This is similar to the BPP with an additional constraint that all components needed for each PCB must fit on trolleys less than or equal to the capacity of an assembly line, e.g., if a PCB is to be built on an assembly line with two CAP machines, which can take 16 trolleys, then all PCB components needed for that PCB must be loaded onto less than or equal to 16 trolleys otherwise the problem will be infeasible.

Mathematical notations for the development of the model are defined in Table~\ref{tab_notations} and domain-specific terminology is defined in Table~\ref{tab_terminology}.

\begin{table}[htb!]
		\centering
		\caption{Notation descriptions}
		\label{tab_notations}
		\begin{tabular}{ll}
			\hline
			\textbf{Symbols} & \textbf{Meaning} \\
			\hline
			\multicolumn{2}{l}{\textbf{Indexes:}}\\
			$c$ & index into types of parts/components, $c=1,2,3,...,C$ \\
			$p$ & index into types of PCBs, $p=1,2,3,...,P$\\
			$t$ & index into trolleys for loading components, $t=1,2,3,...,T$\\
			\hline
			\multicolumn{2}{l}{\textbf{Parameters:}}\\
			$C$ & total number of unique components\\
			$P$ & total number of unique PCBs\\
			$T$ & maximum number of trolleys allowed to manufacture $P$ PCBs\\
			$N$ & trolley capacity, i.e., number of slots on the trolley\\
			$l_p$ & maximum trolley limit for PCB $p$ (depends on capacity of assembly line)\\
			$s_c$ & number of slots needed by component $c$ on a trolley\\
			$S_p$ & set of components needed to build PCB $p$\\
			$\mathbb{M}$ & a very large number\\
			\hline
			\multicolumn{2}{l}{\textbf{Variable:}}\\
			$x_{ct}$ & 1 if component $c$ is assigned to trolley $t$ else 0\\
			\hline
			\multicolumn{2}{l}{\textbf{Auxiliary variables:}}\\
			$y_t$ & 1 if trolley $t$ is used else 0\\
			$z_{pt}$ & 1 if PCB $p$ needs trolley $t$ else 0\\
			\hline
	\end{tabular}
\end{table}

\textbf{Objective:} The objective of the trolley optimisation problem is given below as the minimisation of the sum of all trolleys needed, where $y_t$ is an indicator variable which tells if a trolley $t$ is needed or not, i.e., 1, if needed otherwise 0, and $T$, is maximum number of trolleys allowed to solve the problem. This problem is a kind of assignment problem that decides which component $c$ is loaded onto which trolleys $t$ to minimise the use of trolleys.
\begin{equation}
	\label{eq_obj} \min \quad \sum_{t=1}^T y_t
\end{equation}

\textbf{Constraints:} We extend the BPP to formulate the TOP so the problem includes constraints from BPP (\cite{perron2011operations}), in addition to the problem specific constraints, as discussed below.

(a) Each component $c$ can be on only one trolley $t$ in an assembly line, as given by the following constraint, where $x_{ct}$ is an indicator variable which tells if a component $c$ is assigned to a trolley $t$ or not and takes a value 1 if assigned otherwise takes 0.
\begin{equation}
	\label{eq_c1}
	\sum_{t=1}^T x_{ct} = 1, \quad \forall c.
\end{equation}

(b) Each trolley $t$ has a fixed capacity of $N$ so if the trolley is used then the maximum load on it should not exceed the trolley capacity. So, assuming $s_c$ represents the number of slots needed by component $c$ then the trolley capacity constraint is given below.
\begin{equation}
	\label{eq_c2}
	\sum_{c=1}^C s_c \times x_{ct}  \le N \times y_t, \quad \forall t.
\end{equation}

(c) When a PCB is being built on an assembly line, we cannot add or remove trolleys due to the limited capacity of the assembly line, which is determined by the number of CAP machines in the assembly line. Therefore, there is a maximum limit on the number of trolleys required by each PCB. That means if an assembly line can take a maximum of 16 trolleys then all PCBs built on that line should have a maximum trolley limit of 16. Assuming, $z_{pt}$ is an indicator variable which is 1 if trolley $t$ is needed to build PCB $p$ and $l_p$ is the maximum trolley limit for PCB $p$ then the constraint for maximum trolley limit for the PCB is given below.

\begin{equation}
	\label{eq_c3}
	\sum_{t=1}^T z_{pt} \le l_p, \quad \forall p.
\end{equation}

(d) To ensure that $z_{pt}$ is 1 when trolley $t$ is needed to build PCB $p$ else 0, following constraints are added. This constraint checks if any of the components needed by PCB $p$ is loaded onto trolley $t$ and introduces the following constraints to deal with the `if' statement.
\begin{equation}
	\begin{split}
		\label{eq_c4}
		\sum_{c \in S_p} x_{ct} - \mathbb{M} \times z_{pt} \le 0, \quad \forall t, p,\\
		z_{pt} - \sum_{c \in S_p} x_{ct} \le 0, \quad \forall t, p.
	\end{split}
\end{equation}

The objective function in (\ref{eq_obj}) and constraints (\ref{eq_c1}) and (\ref{eq_c2}) are similar to the BPP (\cite{perron2011operations}). Constraints (\ref{eq_c3}) and (\ref{eq_c4}) extend the BPP to adapt it to the trolley optimisation problem.

The objective function, as well as the constraints, are linear so making the trolley optimisation problem an MILP.

\subsection{Solution approach}
\label{subsec_solution_approach}
Exact optimisation solvers are typically preferred for addressing optimisation problems, particularly when they can deliver optimal solutions within a reasonable timeframe (\cite{chauhan2023exploitation}). Since our problem falls into the medium-size category and exact optimization solvers can efficiently yield solutions, we chose this approach to solve the TOP.

Moreover, since our problem is an MILP, we apply mathematical programming solution techniques as well as constraint programming to solve the TOP. The intuition to use constraint programming was that constraint programming can perform better than MILP techniques when variables are binary, which is the case for the TOP. The experiments also support our intuition as constraint programming takes 30 minutes to find the optimal solution of the TOP as compared with 75 minutes taken by mathematical programming techniques. So, the final results are reported using the constraint programming in Subsec.~\ref{subsec_results}.

\subsection{Dataset and experimental settings}
\label{subsec_dataset}
We present the case study of a company having an assembly shop with multiple assembly lines, which builds a variety of PCBs in low-volume and high-mix on each assembly line. The company needs to load components required to build all PCBs on the assembly line, onto trolleys and stackers. Here, two industrial datasets from the company, corresponding to two assembly lines with two and three CAP machines, respectively, are presented, whose statistics are given in Table~\ref{tab_datasets}. The two datasets, named as `A' and `B', have 80 and 62 PCBs, each of which is built from a subset of 579 and 930 components, respectively. Most of the components, requiring up to five slots are placed onto the trolleys and bigger components are placed on the stackers. The distribution of the size of components for both datasets is given in Fig.~\ref{fig_distri_slots}. From the figure, it is clear that most of the components need one slot and a small number of components are loaded onto the stackers. Since PCBs in dataset A are built based on an assembly line equipped with two CAP machines, the assembly line's capacity (referred to as the maximum trolley limit in the model development) is either 14 trolleys and 2 stackers, or 16 trolleys if stackers are not needed. Consequently, for any PCB in dataset A, the trolley/stacker requirements must not exceed the assembly line's capacity. Similarly, since the PCBs in dataset B are built in an assembly line with three CAP machines so for any PCB in dataset B, the trolley/stacker requirements must not exceed 22 trolleys and 2 stackers or 24 trolleys. We set the maximum number of allowed trolleys equal to the solution obtained using the manual approach, i.e., $T=28$ and 50 for datasets A and B, respectively. This value is very important because a very small number can make the problem infeasible and a very large number will increase the number of variables, thereby the problem will become computationally expensive.

\begin{table}[htb!]
	\centering
	\caption{Details of two industrial datasets}
	\label{tab_datasets}
	\begin{tabular}{|c|c|ccc|c|c|}
		\hline
		\multirow{2}{*}{\textbf{Dataset}} & \multirow{2}{*}{\textbf{\#PCBs}} & \multicolumn{3}{c|}{\textbf{\#Components}}                                            & \multirow{2}{*}{\textbf{\#Variables}} & \multirow{2}{*}{\textbf{\#Constraints}} \\ \cline{3-5}
                                  &                                  & \multicolumn{1}{c|}{Trolley} & \multicolumn{1}{c|}{Stacker} & \textbf{Total} &                       &                              \\ \hline
\textbf{A}                        & 80                               & \multicolumn{1}{c|}{537}     & \multicolumn{1}{c|}{42}      & 579            &                 17,304      &   5,125                           \\ \hline
\textbf{B}                        & 62                               & \multicolumn{1}{c|}{875}     & \multicolumn{1}{c|}{55}      & 930            &                 46,900      &  7,187                            \\ \hline
	\end{tabular}
\end{table}

The total number of variables in the problem is the sum of the primary variables and auxiliary variables and can be calculated as $x_{ct}+y_t+z_{pt},\; \forall c, t, p$, i.e., $C \times T + T + P \times T$. The total number of constraints can be calculated as $C + T + P + 2\times T \times P $. For example, for dataset A, the number of variables for the assignment of components to trolleys is $537 \times 28 + 28 + 80 \times 28 = 17,304$, and the number of constraints is $537 + 28 + 80 + 2\times 28 \times 80 = 5,125$. Similarly, we can calculate the number of variables and the constraints for the assignment of components to the stacker as well as for dataset B which are provided in Table~\ref{tab_datasets}.

\begin{figure}[htb!]
	\centering
	\includegraphics[width=0.6\textwidth]{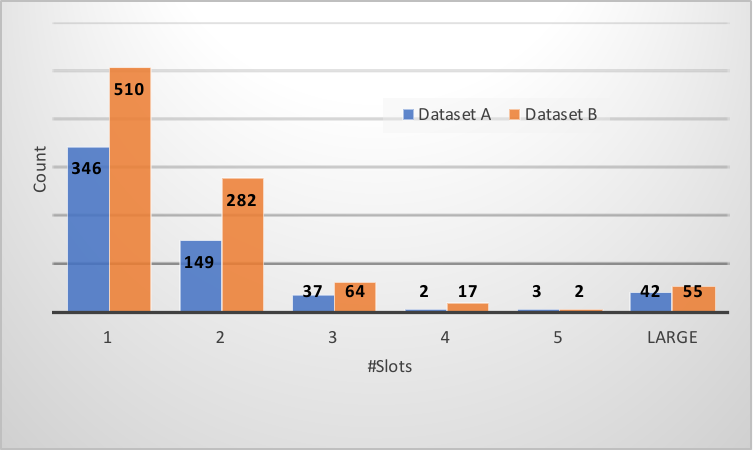}
	\caption{Distribution of component sizes in terms of number of slots required on a trolley/stacker.}
	\label{fig_distri_slots}
\end{figure}

The experiments are coded in Python programming language and use Google OR tools (\cite{perron2011operations}) to solve the TOP. The Google OR-tools library is used because it provides mathematical programming and constraint programming solvers to solve MILP and is freely available. All the experiments are executed on a MacBook Pro (16GB RAM, 256 SSD, 2.5 GHz Dual-Core Intel Core i7).

\subsection{Comparative study}
\label{subsec_results}
We compare the results of the proposed model to the current practice, i.e., the manual process used by the company to solve the TOP. The current practice suffers mainly from two issues. First and the most important issue is the long time of eight weeks to prepare the trolley loading setup which causes the assembly shop to lose the flexibility to introduce new PCBs in the build process and to respond to unprecedented situations, like COVID-19. Second, the manual process is not optimal and requires more trolleys which are difficult to manage.

\begin{figure*}[htb!]
	\centering
	\begin{subfigure}{0.4\textwidth}
		\centering
		\includegraphics[width=\textwidth]{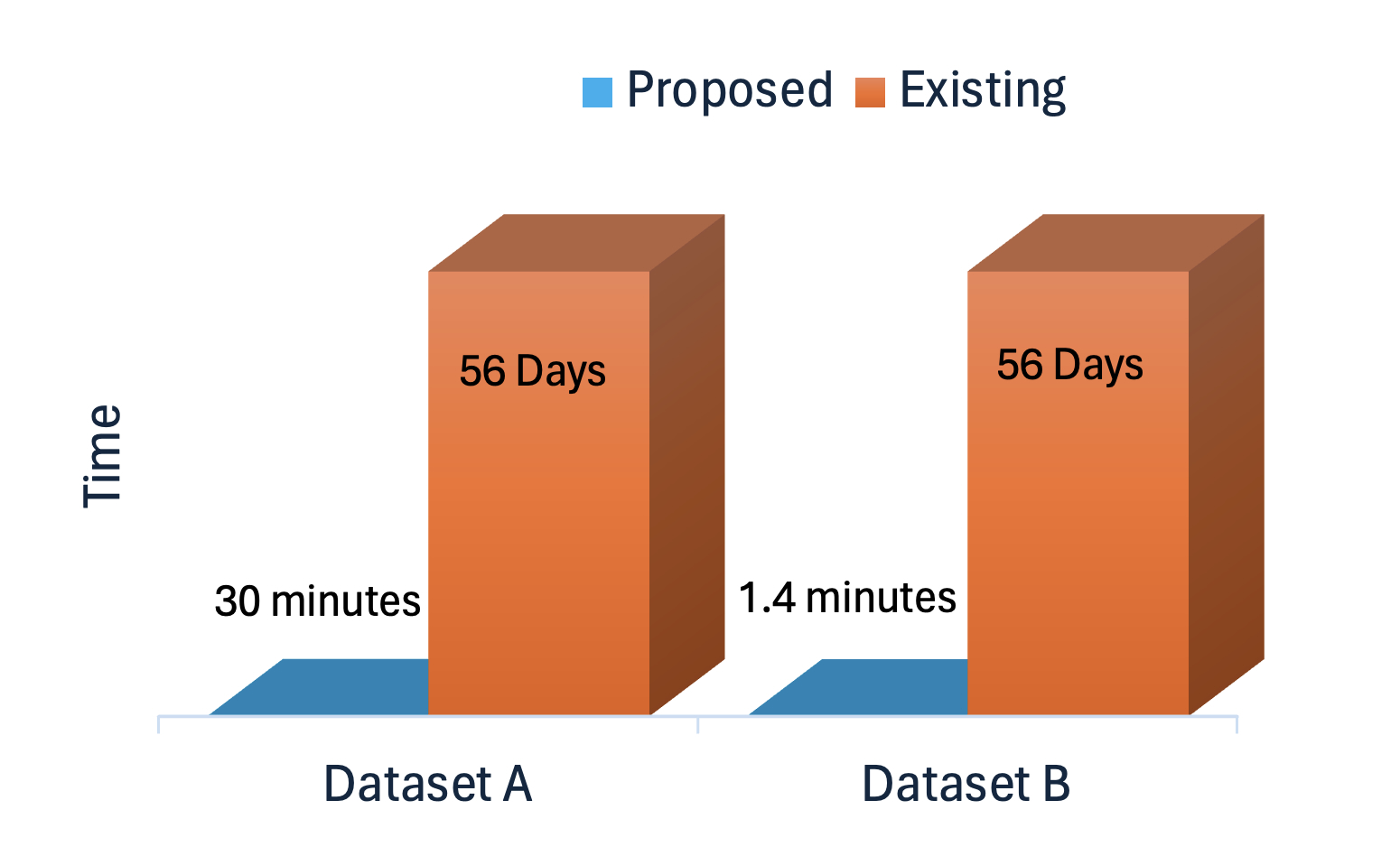}
		\caption{Time to solve the problem}
		\label{subfig_time}
	\end{subfigure}%
	~~
	\begin{subfigure}{0.4\textwidth}
		\centering
		\includegraphics[width=\textwidth]{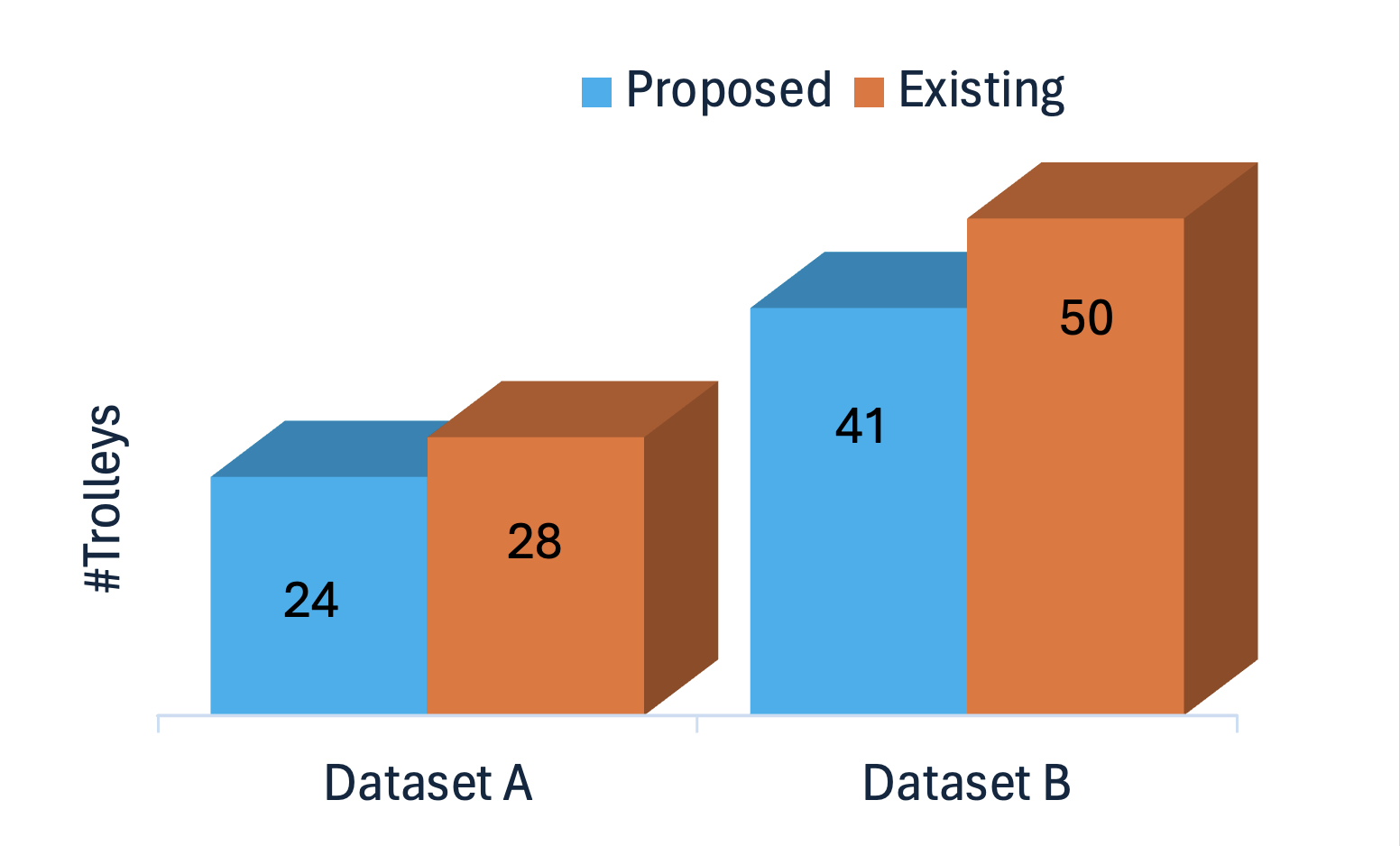}
		\caption{Number of trolleys needed}
		\label{subfig_no_trolley}
	\end{subfigure}	
	\caption{Comparative study of the existing manual method versus the proposed automated TOP method}
	\label{fig_results}
\end{figure*}

The comparative study is presented in Fig.~\ref{fig_results}, which compares the time to solve the problem and the number of trolleys needed to load the given number of components to build a given set of PCBs. As evident from Subfig.~\ref{subfig_time}, the proposed model exhibits significant improvement over the manual process. It can efficiently solve the problem in less than one hour, whereas the manual process requires 56 days for both datasets. It is noted that the model takes longer on dataset A than B, despite having a larger number of components in dataset B. This is because of a more relaxed capacity constraint for the assembly line for dataset B as each PCB can take 22 or 24 trolleys compared with 14 or 16 trolleys taken by each PCB in dataset A (please refer to Subsec.~\ref{subsec_sensitivity} for sensitivity analysis on maximum trolley limit). In addition to that, as it is clear from the Subfig.~\ref{subfig_no_trolley}, the proposed model also loads the given components onto a lesser and optimal number of trolleys as compared with the number of trolleys taken in the manual process to load the same set of components. Therefore, the proposed model effectively automates the TOP, addressing the associated issues of prolonged processing time and inflexibility on the shop floor.

\begin{figure*}[htb!]
	\centering
	\begin{subfigure}{0.5\textwidth}
		\centering
		\includegraphics[width=\textwidth]{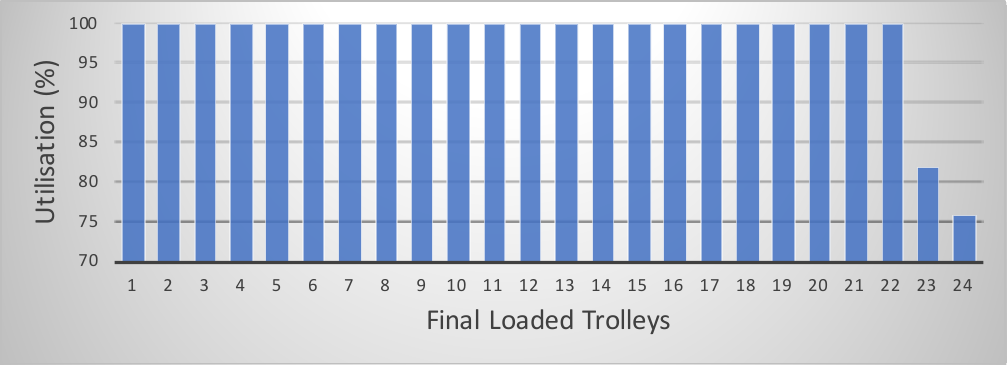}
		\label{fig_final_trolleys1}
	\end{subfigure}%
	~~
	\begin{subfigure}{0.5\textwidth}			
		\centering
		\includegraphics[width=\textwidth]{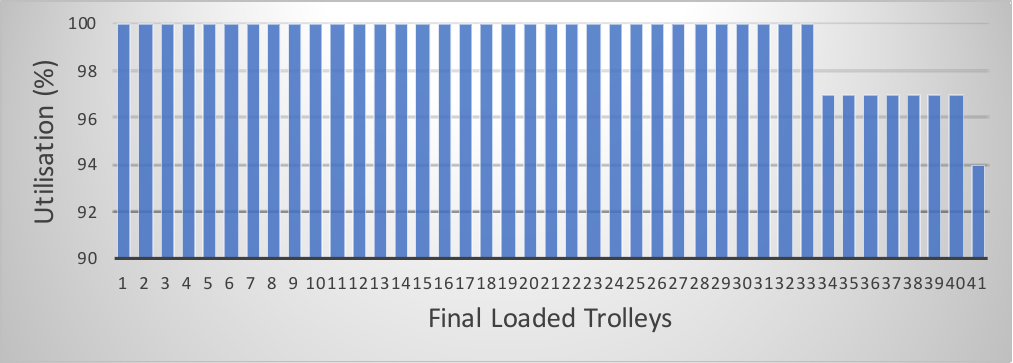}
		\label{fig_final_trolleys2}
	\end{subfigure}	
	\caption{Utilisation status of selected trolleys on (a) Dataset A and (b) Dataset B}
	\label{fig_final_trolleys}
\end{figure*}

The utilisation of the capacity of each of the selected trolleys is represented in Fig.~\ref{fig_final_trolleys}. As it is clear from the figure, the model can utilise most of the capacity of the trolleys for both datasets as only a few trolleys have some unused spaces. However, the utilisation of different trolleys is dependent on the problem. In both datasets, the PCBs belong to certain product families with similar component requirements, likely contributing to the high utilisation of trolleys in the final solution.

The problem of component assignment to stackers is much smaller as compared to the problem of component assignment to trolleys and the model takes a few seconds to solve the problem and loads the components onto two stackers for both datasets, which is also equal to the number of stackers taken in the manual process. Thus, overall, both the assignment problems, i.e., the TOP can be solved in less than one hour and takes 24 trolleys and 2 stackers for dataset A, and 41 trolleys and 2 stackers for dataset B.

The proposed model has been successfully deployed in the company, leading to substantial cost reductions through the automation of the TOP. Additionally, it offers flexibility in rescheduling, introducing new products, and addressing unforeseen circumstances such as the COVID-19 pandemic in the build process.

\subsection{Sensitivity analysis}
\label{subsec_sensitivity}
In this subsection, we discuss the sensitivity of the model to the `maximum trolley limit' constraint for each PCB on dataset A. We solve the model for three more settings, in addition to the case study. For the simplicity of the experiments, first, we set the same limits for all the PCBs as 16, 20 and 22 trolleys. From the experiments, we find that there is no change in the use of the number of trolleys used to load the components although their specific loading does change. However, there is a large change in the time to solve the problem as the model takes 30 minutes (m), 2 m, 0.35 m, 0.18 and 0.18 m for the case study setting, 16, 18, 20 and 22 trolleys, respectively. This large difference in time to solve the problem is because an increase in the maximum trolley limit makes the problem less and less restrictive until the maximum trolley limit has no effect. Since 20 and 22 trolleys take almost the same time so the problem has no effect after 20 trolleys. Thus, we find that the problem is sensitive to the maximum trolley limit constraint and is one of the reasons for the better performance of the model on dataset B.

\section{Trolley optimisation is an NP-complete problem}
\label{sec_np_complete}
Here, we prove the NP-completeness of the TOP. Our proof is motivated by \cite{BPP}. We consider only trolleys and ignore stackers because as discussed above, the problem can be decomposed into smaller sub-problems and a single model can solve the assignment of trolleys as well as stackers. First, we define the decision version of the problem.

\begin{definition}[Trolley optimisation]
We are given a set $\sigma = \lbrace 1,2,..,c,..,C\rbrace$ of PCB components where each component $c$ has size $s_c$, a set $\tau = \lbrace 1,2,..,t,..,T\rbrace$ of trolleys of a fixed size of $N$ and a set $\rho = \lbrace 1,2,...,p,...,P \rbrace$ of PCBs where each PCB $p$ needs a subset of components $\sigma_p \subseteq \sigma$ to build the PCB such that the components in $\sigma_p$ can be loaded in maximum of $v_p$ trolleys. Decide if PCB components can be loaded in less than or equal to $k \in \mathbb{Z^+}$ number of trolleys while obeying the constraint on a maximum number of trolleys needed for each PCB.
\end{definition}

Next, we provide the proof for NP-completeness of the TOP.
\begin{theorem}
\label{np_complete}
The trolley optimisation is an NP-complete problem.
\end{theorem}

\begin{proof}
A problem `A' is an NP-complete if there exists a known NP-complete problem `B', such that `B' can be reduced to `A'. We know that the bin packing problem is a well-known NP-complete problem [\cite{BPP}]. So, to prove the NP-completeness of the TOP, we reduce the BPP to the TOP. Let's recall the decision version of the BPP: We are given a set $\sigma = \lbrace 1,2,..,c,..,C\rbrace$ of items where each item $c$ has size $s_c$ and a set $\tau = \lbrace 1,2,..,t,..,T\rbrace$ of bins of a fixed size of $N$. Decide if items can be packed in less than or equal to $k$ number of bins. Now suppose $v_p = k$, i.e., each PCB $p$ can use all the available trolleys, and then the TOP becomes equivalent to the BPP. Thus, a set of components $\sigma$ required to build a set of PCBs $\rho$ can be loaded on $k$ trolleys \textit{if and only if} components can be packed in $k \in \mathbb{Z^+}$ trolleys.
\end{proof}  
Hence, the theorem is proved.

\section{Conclusion and discussions}
\label{sec_conclusion}
PCB manufacturing is an important activity due to the wide use of PCBs in all electrical and electronic equipment, like radio, television, smartphones, computers, aircraft engines and trains etc. Recently, there has been a growing need for PCBs due to the advancements in technology, like digitalisation and adoption of Internet of Things (IoT) devices etc. So, the manufacturers use automated and complex manufacturing systems, like CAP machines, which need complex setup operations.

In this paper, we introduce a novel NP-complete problem of trolley optimisation, i.e., the problem of assignment of a given set of PCB components to trolleys and stackers for the production of a set of PCBs in an assembly line. The earlier research mainly focused on other PCB planning problems, like PCB allocations, component placement sequences etc. and assumed direct assignment of components to CAP machines. The direct assignments are not practical, especially for a large variety of PCBs, as industries, generally, use trolleys for easy management of components and for ease of changing trolleys between different PCBs. So, this work automates the TOP to build a given set of PCBs which mainly focuses on reducing the time to prepare trolleys and on reducing the number of required trolleys.

We present a novel extension of the BPP to formulate the TOP. Similar to bin packing, the TOP finds a minimum number of trolleys/stackers (equivalent to bins) of common capacity to load/pack a given set of components (equivalent to items) of different sizes/weights to build a set of PCBs in an assembly line. The TOP shares the objective function and constraints of the BPP but adds additional constraints to ensure that each PCB is feasible on the assembly line, i.e., the problem limits the maximum number of trolleys to be used to load components of a PCB otherwise either the PCB could not be built on the assembly line or need to replace trolleys during the building process. Additionally, we exploit the problem structure to decompose the TOP into two smaller, identical and independent problems, by pre-computing the dependency between them. So, a single and a smaller MILP model is developed to solve the TOP which is solved using exact optimisation methods. We also proved that the TOP is an NP-complete problem.

We present a case study of an aerospace company which has an assembly shop to meet its PCB needs across all programmes and operates in a low-volume and high-mix setting. Due to the complex nature of the PCB planning problem and due to the customised needs of the company, all the setups are performed manually by experienced managers. The manual setup of trolley loading for each assembly line takes eight weeks and results in loss of flexibility on the shop floor. The proposed MILP model successfully solves the TOP in less than an hour and is deployed in the company. The model resulted in significant cost reductions through automation of the TOP, in addition to introducing flexibility in the building process.

From the managerial perspective, the TOP model not only saves money by automating the manual process but also makes the task easy to manage with fewer trolleys and brings flexibility in the planning of different builds and in satisfying different demands. This enables the managers to easily introduce new products, due to the ability of the model to generate faster automated setups, adapt to changing demands and to quickly respond to unprecedented situations, like COVID-19.

As discussed earlier, PCB planning is a multi-level optimisation problem with several interdependent sub-problems, which are complex due to the NP-hard nature of the problem and the scale of the problem. So, in general, heuristic-based methods are used to solve them either independently or some of them. In this work, we have used the exact optimisation method to solve one of the sub-problems which obviously won't result in an optimal solution for the integrated PCB planning problem, and presents potential avenues for future research. Our future work will study the TOP with other PCB planning sub-problems, especially the scheduling of PCBs in an assembly shop with multiple assembly lines.

\section*{Acknowledgement}
This research was funded by Aerospace Technology Institute and Innovate UK, the UK’s innovation funding agency, through the ``Digitally Optimised Through-Life Engineering Services" project (113174). We are also grateful to the anonymous reviewers for their constructive comments, which greatly helped to improve the clarity of our work.

\section*{Statements and Declarations}
\subsection*{Disclosure statement}
The authors report there are no competing interests to declare.

\subsection*{Data availability statement}
Due to the commercial nature of this research, the supporting data is unavailable.

\subsection*{CRediT authorship contribution statement}
Vinod Kumar Chauhan: Conceptualisation, Methodology, Software, Writing – original draft. Mark Bass: Conceptualisation. Ajith Kumar Parlikad: Supervision, Funding acquisition. Alexandra Brintrup: Conceptualisation, Supervision, Funding acquisition. All authors reviewed the final draft.
	
\bibliography{pcb} 

\end{document}